\documentclass[12pt]{article} 
\usepackage{ctable}
\usepackage{latexsym}
\usepackage{amsmath}
\usepackage{amssymb}
\usepackage{epsfig}
\usepackage{latexsym}
\usepackage{setspace}
\usepackage{multirow}
\usepackage[pdftex,pdfborder={0 0 0 0}]{hyperref}
\usepackage{relsize}
\usepackage{comment}
\usepackage{enumerate}

\usepackage[margin = 1.25in]{geometry}
\usepackage[longnamesfirst]{natbib}
\usepackage[mdyyyy]{datetime}
\def\E{{\mathbb E}}        
\def\var{{\text{Var}}}

\def\cov{{\text{cov}}}

\newcommand{\nn}{\nonumber\\}
\newtheorem{theorem}{Theorem}
\newtheorem{lemma}[theorem]{Lemma}

\newenvironment{proof}{{\bf Proof:}}{\hfill\rule{2mm}{2mm} \par}

\usepackage[scaled=0.92]{helvet}

\begin{document}
\title{Blocking estimators and inference under the Neyman-Rubin model} 
\author{Michael J.~Higgins\footnote{Assistant Professor, Kansas State University.} \and
    Fredrik S\"avje\footnote{Postdoctoral fellow, Department of Political Science and Department of Statistics, UC Berkeley.} \and
    Jasjeet S.~Sekhon\footnote{Robson Professor of Political Science and Statistics, UC Berkeley.}}

\date{\today}

\doublespace
\maketitle
\begin{abstract}
  We derive the variances of estimators for sample
  average treatment effects under the Neyman-Rubin potential outcomes
  model for arbitrary blocking assignments and an arbitrary number of
  treatments.
\end{abstract}

\section{Introduction}

Going back to~\citet{fisher26}, the canonical experimental design for preventing imbalances in covariates is \emph{blocking}: where one groups similar experimental units together and assigns treatment in fixed proportions within groups and independently across groups. If blocking reduces imbalances in prognostically important covariates, blocking will improve the expected precision of estimates. As a consequence of the higher accuracy, variance estimators which do not take the design into account will generally overestimate the uncertainty. \citet{neyman35} discusses estimation with a block design under his potential outcomes model---which is often called the Neyman-Rubin causal model~\citep[NRCM,][]{splawa1990application,rubin1974estimating,holland1986statistics}.  In more recent work, variance for the matched-pairs design under NRCM has been analyzed by~\citet[looking at conditional average treatment effects, CATE, as estimands]{abadie07} and~\citet[population average treatment effects, PATE]{imai08}. \cite{imbens11} provides an overview, and discusses PATE, CATE, and the sample average treatment effect (SATE). He also notes the impossibility of variance estimation of the conditional average treatment effect (CATE) with a matched-pairs design.
  
We build upon this literature and discuss estimation of the sample average treatment effect (SATE) under block designs with arbitrarily many treatment categories and an arbitrary blocking of units. Specially, we derive variances for two unbiased estimators of the SATE---the difference-in-means estimator and the Horvitz-Thompson estimator---and give conservative estimators of these variances whenever block sizes are at least twice the number of treatment categories. Unlike most other methods, these are closed-form formulas of the SATE variance under NRCM applicable to a wide range of settings without making additional parametric assumptions.

\section{Inference of blocking estimators under Neyman-Rubin model \label{secInf}}
\subsection{Notation and preliminaries}
    
      There are $n$ units, numbered $1$ through $n$.
      There are $r$ treatments, numbered 1 through $r$.
      Each unit is assigned to exactly one treatment.
      Each unit $i$ has a vector of block covariates $\mathbf{x_i}$.      
      A distance between block covariates (such as the Mahalanobis distance) 
      can be computed between each pair of distinct covariates.
      
      Suppose the units are partitioned into $b$ blocks,
      numbered 1 through $b$,
      with each block containing at least $t^*$ units, with $t^* \geq r$.
      Let $n_c$ denote the number of units in block $c$.
      Assume that the units within each block $c$ are ordered in some way:
      let $(k,c)$ denote the $k^{\text{th}}$ unit in block $c$.
      Let $z$ denote the remainder of $n/r$, and let
      $z_c$ denote the remainder of $n_c/r$.

      \subsubsection{Balanced complete and block randomization}
      Treatment assignment is 
      \textit{balanced} if
          $z$ 
          treatments are replicated $\lfloor n/r\rfloor + 1$ times,
          and $r-z$ of the treatments are replicated $\lfloor n/r\rfloor$ times.
      A balanced treatment assignment is 
      \textit{completely randomized} if each of the
            \begin{equation}
              \binom r z\prod_{i = 0}^{z-1} \binom{n - i(\lfloor n/r \rfloor + 1)}{\lfloor n/r \rfloor + 1}
              \prod_{i = 0}^{r-z-1} \binom{n - z(\lfloor n/r \rfloor + 1) - i\lfloor n/r \rfloor}{\lfloor n/r \rfloor }
            \end{equation}
          possible treatment assignments are equally likely.
      Treatment is \textit{balanced block randomized} if treatment is 
      balanced and completely randomized within
      each block and treatment is assigned independently across blocks.
      
      Let $T_{kcs}$ denote treatment indicators for each unit $(k,c)$:
      \begin{equation}
        T_{kcs} = \left\{
        \begin{array}{ll}
          1, & \text{unit $(k,c)$ receives treatment $s$},\\
          0, &\text{otherwise.}
        \end{array}
        \right.
      \end{equation}
      Let $\#T_{cs} = \sum_{k = 1}^{n_c} T_{kcs}$ denote the 
      number of units in block $c$ that
      receive treatment $s$, and let
      $\#T_{s} = \sum_{c=1}^b\#T_{cs}$ denote
      the number of units in total assigned to $s$.
      Let $z_c$ denote the remainder of $n_c/r$.   
      
     Under balanced complete randomization, $\#T_{s}$ has distribution
      \begin{equation}
        \#T_{s} = \left\{ 
          \begin{array}{ll}
            \lfloor n/r \rfloor + 1& \text{ with probability } z/r,\\
            \lfloor n/r \rfloor& \text{ with probability } (r - z)/r.
          \end{array}
        \right.
        \label{distnums}
      \end{equation}
      Under balanced block randomization, $\#T_{cs}$ has distribution
      \begin{equation}
        \#T_{cs} = \left\{ 
          \begin{array}{ll}
            \lfloor n_c/r \rfloor + 1& \text{ with probability } z_c/r,\\
            \lfloor n_c/r \rfloor& \text{ with probability } (r - z_c)/r.
          \end{array}
        \right.
        \label{distnumcs}
      \end{equation}
      Since $t^* \geq r$, it follows that $\#T_{s} \geq \#T_{cs} \geq 1$.

      \subsubsection{Model for response: the Neyman-Rubin Causal Model}

      We assume responses follow
      the Neyman-Rubin Causal Model 
      (NRCM)~\citep{splawa1990application, rubin1974estimating, holland1986statistics}.
      Let $y_{kcs}$ denote the \textit{potential outcome} of unit $(k,c)$ given 
      treatment $s$---the hypothetical observed value of unit $(k,c)$ 
      had that unit received treatment $s$.
      Under the NRCM, the potential outcome $y_{kcs}$ is non-random, and
      the value of this outcome is observed if and only if $(k,c)$ receives treatment $s$;
      exactly one of $\{y_{kcs}\}_{s = 1}^r$ is observed.
      The observed response is:
      \begin{equation}
        Y_{kc} \equiv y_{kc1}T_{kc1} + y_{kc2}T_{kc2} + \cdots + y_{kcr}T_{kcr}.
      \end{equation}
      Inherent in this equation is the \textit{stable-unit treatment value assumption} (SUTVA):
      the observed $Y_{kc}$ only depends on which treatment is assigned to unit $(k,c)$,
      and is not affected by the treatment assignment of any other unit $(k',c')$.
      	
    \subsubsection{Common parameters and estimates under the Neyman-Rubin Causal Model
    \label{commonparsubsec}}
   
      The domain-level mean and variance of potential outcomes for treatment $s$ are:
      \begin{align}
       \mu_{s} & \equiv \frac{1}{n}\sum_{c=1}^b\sum_{k=1}^{n_c} y_{kcs},\label{definedomainmean}\\
        \sigma_{s}^2 &\equiv \sum_{c=1}^b\sum_{k=1}^{n_c}\frac{(y_{kcs} - \mu_{s})^2}{n}  =
          \sum_{c=1}^b\sum_{k=1}^{n_c} \frac{y_{kcs}^2}{n} - \left(\sum_{c=1}^b\sum_{k=1}^{n_c} \frac{y_{kcs}}{n}\right)^2,\label{definedomainvar}
          \intertext{and the domain-level covariance between potential outcomes for treatment $s$ and treatment $t$ is:}
        \gamma_{st} &\equiv  \sum_{c=1}^b\sum_{k=1}^{n_c}\frac{(y_{kcs} - \mu_{s})(y_{kct} - \mu_{t})}{n} \nn &=
           \sum_{c=1}^b\sum_{k=1}^{n_c}\frac{y_{kcs}y_{kct}}{n} 
          -  \left(\sum_{c=1}^b\sum_{k=1}^{n_c} \frac{y_{kcs}}{n}\right) \left(\sum_{c=1}^b\sum_{k=1}^{n_c} \frac{y_{kct}}{n}\right).\label{definedomaincov}
       \end{align}
    
    Two estimators for $\mu_s$ 
    are the sample mean and the Horvitz-Thompson estimator~\citep{horvitz1952generalization}:
    \begin{align}
       \hat\mu_{s,\text{samp}} &\equiv \sum_{c=1}^b\sum_{k=1}^{n_c} \frac{y_{kcs}T_{kcs}}{\#T_{s}},\\
       \hat\mu_{s,\text{HT}}  &\equiv \sum_{c=1}^b\sum_{k=1}^{n_c} \frac{y_{kcs}T_{kcs}}{n/r}. \\
   \intertext{Two estimators for $\sigma_s^2$ are the sample variance and the
   Horvitz-Thompson estimate of the variance:}
      \hat\sigma^2_{s,\text{samp}} &\equiv \frac{n-1}{n}\sum_{c=1}^b\sum_{k=1}^{n_c}\frac{T_{kcs}\left(y_{kcs} - \sum_{c=1}^b\sum_{k=1}^{n_c}\frac{y_{kcs}T_{kcs}}{\#T_{s}}\right)^2}{\#T_{s} -1}.\\
      \hat\sigma^2_{s,\text{HT}} &\equiv \frac{(n-1)r}{n^2}\sum_{c=1}^b\sum_{k=1}^{n_c} y^2_{kcs}T_{kcs}  \nn & ~~~~~
      -\frac {(n-1)r^2} {n^2(n-r) + nz(r-z)}\sum_{(k,c)\neq (k',c')} y_{kcs}y_{k'c's}T_{kcs}T_{k'c's}
    \end{align}
    The sample estimators weight observations by the inverse of the number of observations receiving treatment $s$, 
    and the Horvitz-Thompson estimators weight observations by 
    the inverse of the probability of being assigned treatment $s$.     
    Block-level parameters $\mu_{cs}$,  $\sigma_{cs}$, and $\gamma_{cst}$,
      and block-level estimators $\hat\mu_{cs,\text{samp}}$, $\hat\mu_{cs,\text{HT}}$,
     $\hat\sigma^2_{cs,\text{samp}}$, and $\hat\sigma^2_{cs,\text{HT}}$
      are defined as above except that sums range over only units in block $c$.    
      
    When treatment is balanced and completely randomized, the domain-level
    estimators satisfy the following properties:
    \begin{lemma}
        Under balanced and completely randomized treatment assignment, 
        for any treatments $s$ and $t$ with $s \neq t$,
        \begin{eqnarray}
          \E(\hat\mu_{s,\text{samp}}) &=& \mu_{s},\\
          \var(\hat\mu_{s,\text{samp}}) &=& \frac{r-1}{n-1}\sigma_s^2 + \frac{rz(r-z)}{(n-1)(n-z)(n+r-z)}\sigma_s^2 , \\
          \cov(\hat\mu_{s,\text{samp}},\hat\mu_{t,\text{samp}}) &=& \frac{-\gamma_{st}}{n-1}.
        \end{eqnarray}
        \label{lemmadiffblockests}
      \end{lemma}
      
      \pagebreak
      \begin{lemma}
        Under balanced and completely randomized treatment assignment, 
        for any treatments $s$ and $t$ with $s \neq t$,
        \begin{eqnarray}
          \E(\hat\mu_{s,\text{HT}}) &=& \mu_{s},\\
          \var(\hat\mu_{s,\text{HT}}) &=& \frac{r-1}{n-1}\sigma_s^2  + \frac{z(r-z)}{n^3(n-1)}\sum_{(k,c) \neq (k',c')}y_{kcs}y_{k' c's}, \nn \\
          \cov(\hat\mu_{s,\text{HT}},\hat\mu_{t,\text{HT}}) &=& \frac{-\gamma_{st}}{n-1}
        - \frac{z(r-z)}{(r-1)n^3(n-1)}\sum_{(k,c) \neq (k',c')}y_{kcs}y_{k' ct}.
        \end{eqnarray}
        \label{lemmahtblockests}
      \end{lemma}
      \begin{lemma}
      Under balanced and completely randomized treatment assignment, for any treatment $s$,
        \begin{equation}
          \E(\hat\sigma^2_{s,\text{samp}} ) = \sigma^2_{s}, ~~~~~
          \E(\hat\sigma^2_{s,\text{HT}} ) = \sigma^2_{s}.
        \end{equation}
        \label{lemmaexpsampvar}
      \end{lemma}
      Recall that, in balanced block randomized designs, treatment is 
      balanced and completely randomized 
      within each block.
      Thus, analogous properties for block-level estimators hold under balanced block randomization.
      Lemmas~\ref{lemmadiffblockests} and~\ref{lemmahtblockests} are proven in 
      Appendix~\ref{proofoflemmasforvarcals}, and Lemma~\ref{lemmaexpsampvar}
      is proven in Appendix~\ref{proofblocklevelvarests}.
    
        Under the NRCM, the covariance $\gamma_{st}$ is not directly estimable;
        such an estimate would require 
        knowledge of potential outcomes under both
        treatment $s$ and treatment $t$ within a single unit.
        However, when blocks contain several replications of
        each treatment, and when potential outcomes satisfy some
        smoothness conditions with respect to the block covariates,
        good estimates of the block-level covariances $\gamma_{cst}$ may be
        obtained.  
        For details, see~\citet{abadie07, imbens11}.
        
    \subsection{Estimating the sample average treatment effect}
      Given any two treatments $s$ and $t$, we wish to estimate the
      \textit{sample average treatment effect of
      $s$ relative to $t$} (SATE$_{st}$),
      denoted $\delta_{st}$.
      The SATE$_{st}$ is a sum of differences of potential outcomes:
      \begin{equation}
        \delta_{st} \equiv \frac 1 n \sum_{c=1}^b\sum_{k = 1}^{n_c}  (y_{kcs} - y_{kct})
        = \sum_{c=1}^b \frac{n_c}{n}\sum_{k = 1}^{n_c}\left(\mu_{cs} -\mu_{ct}\right).
      \end{equation}
      We consider two estimators for $\delta_{st}$:
      the difference-in-means estimator:
      \begin{equation}
        \hat\delta_{st,\text{diff}} \equiv \sum_{c=1}^b\frac {n_c}{n}\sum_{k=1}^{n_c} 
        \left( 
          \frac{y_{kcs}T_{kcs}}{\#T_{cs}} - \frac{y_{kct}T_{kct}}{\# T_{ct}}
        \right)= \sum_{c=1}^b\frac {n_c}{n}\sum_{k=1}^{n_c} 
        \left( 
         \hat\mu_{cs,\text{samp}} - \hat\mu_{ct,\text{samp}}
        \right)
      \end{equation}
      and the Horvitz-Thompson estimator~\citep{horvitz1952generalization}:
      \begin{equation}
        \hat\delta_{st,\text{HT}} \equiv \sum_{c=1}^b\frac {n_c}{n}\sum_{k=1}^{n_c} 
        \left( 
          \frac{y_{kcs}T_{kcs}}{n_c/r} - \frac{y_{kct}T_{kct}}{n_c/r}
        \right)= 
        \sum_{c=1}^b\frac {n_c}{n}\sum_{k=1}^{n_c} 
        \left( 
         \hat\mu_{cs,\text{HT}} - \hat\mu_{ct,\text{HT}}
        \right).
      \end{equation}
      These estimators are shown to be unbiased under
      balanced block randomization in Theorems~\ref{brdiffthm} and~\ref{brhtthm}.            
      
      Properties of these estimators are most easily seen by 
      analyzing the block-level terms.
      Consider first the difference-in-means estimator.
      By linearity of expectations:
       \begin{align}
         \E(\hat\delta_{st, \text{diff}}) &= 
         \sum_{c=1}^b\frac{n_c}n(\E(\hat\mu_{cs,\text{samp}}) - \E(\hat\mu_{ct,\text{samp}})).\label{myexpdiff}\\
         \intertext{When treatment is balanced block randomized, 
         by independence of treatment assignment 
      across blocks:}
       \var(\hat\delta_{st,\text{diff}}) &=
         \sum_{c=1}^b\frac{n^2_c}{n^2}\left[\var(\hat\mu_{cs,\text{samp}}) + \var(\hat\mu_{ct,\text{samp}})
         -2\cov(\hat\mu_{cs,\text{samp}},\hat\mu_{ct,\text{samp}})\right].\label{myvardiff}
	\end{align}
	Linearity of expectations and independence across blocks can also be exploited to
	obtain similar expressions hold for the Horvitz-Thompson estimator.
      
      From Lemmas~\ref{lemmadiffblockests} and~\ref{lemmahtblockests},
      and using~\eqref{myexpdiff} and~\eqref{myvardiff}, 
      we can show that both the difference-in-means estimator and the
      Horvitz-Thompson estimator for the SATE$_{st}$ are unbiased, 
      and we can compute the variance of these estimates.
      \begin{theorem}
        Under balanced block randomization, for any treatments $s$ and $t$ with $s \neq t$:
        \begin{eqnarray}
          \E(\hat\delta_{st,\text{diff}}) &=& \delta_{st},\\
          \var(\hat\delta_{st,\text{diff}}) &=& \sum_{c=1}^b \frac{n^2_c}{n^2}
           \left(\frac{r-1}{n_c-1}(\sigma_{cs}^2+\sigma_{ct}^2)
         + 2\frac{\gamma_{cst}}{n_c-1}\right) \nn&&+ \sum_{c=1}^b \frac{n^2_c}{n^2}
           \left(\frac{rz_c(r-z_c)}{(n_c-1)(n_c-z_c)(n_c+r-z_c)}(\sigma^2_{cs} + \sigma^2_{ct})\right).
        \end{eqnarray}
        \label{brdiffthm}
      \end{theorem}
      \begin{theorem}
        Under balanced block randomization, for any treatments $s$ and $t$ with $s \neq t$:
        \begin{eqnarray}
          \E(\hat\delta_{st,\text{HT}}) &=& \delta_{st},\\
          \var(\hat\delta_{st,\text{HT}}) &=& 
          \sum_{c=1}^b\frac{n^2_c}{n^2}
            \left( \frac{r-1 }{n_c-1}(\sigma_{cs}^2+\sigma_{ct}^2) + 2\frac{\gamma_{cst}}{n_c-1}\right)\nn&&
            +  \sum_{c=1}^b\frac{z_c(r-z_c)}{n_c^3(r-1)}\sum_{k=1}^{n_c}\sum_{\ell \neq k}\left(
              \frac{r-1}{n_c -1}\left( y_{kcs}y_{\ell cs} + y_{kct}y_{\ell ct}\right) + 2\frac{y_{kcs}y_{\ell ct}}{n_c-1}\right)
             \nn
        \end{eqnarray}
                \label{brhtthm}
      \end{theorem}
   
      Note that, when $r$ divides each $n_c$,
      then $\hat\delta_{st,\text{diff}} =  \hat\delta_{st,\text{HT}}$ and       
      \begin{eqnarray}
        && \var(\hat\delta_{st,\text{diff}}) =  \var(\hat\delta_{st,\text{HT}}) =
         \sum_{c=1}^b \frac{n^2_c}{n^2}
           \left(\frac{r-1}{n_c-1}(\sigma_{cs}^2+\sigma_{ct}^2)
         + 2\frac{\gamma_{cst}}{n_c-1}\right)\!\!.
         \label{varwhenrdividesnc}
        \end{eqnarray}
	When $r$ does not divide each $n_c$, simulation results (not
        presented) seem to suggest that the difference-in-means
        estimator has a smaller variance, especially when block sizes
        are small.
        
      \subsection{Estimating the variance}
        As discussed in Section~\ref{commonparsubsec}, 
        estimation of the variance for both the difference-in-means and
        Horvitz-Thompson estimators is complicated by
        $\gamma_{cst}$ terms, which cannot be estimated without making
        assumptions about the distribution of potential outcomes.
        We give conservative estimates (in expectation) for these variances by first
        deriving unbiased estimators for the block-level variances
        $\var(\hat\mu_{cs,\text{diff}})$ and $\var(\hat\mu_{cs,\text{HT}})$
        and bounding the total variance using the Cauchy-Schwarz inequality and the
        arithmetic mean/geometric mean (AM-GM) inequality~\citep{hardyEtal52}
        on the covariance terms $\gamma$.
        These conservative variances make no distributional assumptions
        on the potential outcomes.
       
      \subsubsection{Block-level variance estimates \label{blocklevelvarsec}}  
      The variance for the block-level estimators (as derived in
      Lemmas~\ref{lemmadiffblockests} and~\ref{lemmahtblockests}) can be estimated 
      unbiasedly several ways.
      We consider the following variance estimators.
            \begin{lemma}
            Define:

            \begin{eqnarray}
        \widehat{ \var}(\hat\mu_{cs,\text{samp}})  &\equiv&  \left( \frac{r-1}{n_c-1} + \frac{rz_c(r-z_c)}{(n_c-1)(n_c-z_c)(n_c+r-z_c)}\right)\hat\sigma^2_{cs, \text{samp}},\\
          \widehat{ \var}(\hat\mu_{cs,\text{HT}}) &\equiv& \frac{r-1}{n_c-1}\hat\sigma_{cs,\text{HT}}^2  \nn&&+ \frac{r^2z_c(r-z_c)}{n_c^3(n_c-r) + n_c^2z_c(r-z_c)}\sum_{k=1}^{n_c}\sum_{\ell\neq k} y_{kcs}y_{\ell cs}T_{kcs}T_{\ell cs},\nn 
           \end{eqnarray}
        Under balanced block randomization, for any treatment $s$:
        \begin{eqnarray}
          \E\left[\widehat{ \var}(\hat\mu_{cs,\text{samp}})\right] = \var(\hat\mu_{cs,\text{samp}}),~~~~
          \E\left[\widehat{ \var}(\hat\mu_{cs,\text{HT}}))\right] = \var(\hat\mu_{cs,\text{HT}}).
        \end{eqnarray}
        \label{blocklevelvarests}
      \end{lemma}
      This Lemma is proven in Appendix~\ref{proofblocklevelvarests}.
      
      \subsubsection{Conservative variance estimates for SATE$_{st}$ estimators}
      
      Define the following variance estimators:
      \begin{eqnarray}
          \widehat\var(\hat\delta_{st,\text{diff}}) &\equiv& \sum_{c=1}^b \frac{n_c^2}{n^2}\left[ \left( \frac{r}{n_c-1} + \frac{rz_c(r-z_c)}{(n_c-1)(n_c-z_c)(n_c+r-z_c)}\right)(\hat\sigma^2_{cs, \text{samp}}+ \hat\sigma^2_{ct, \text{samp}})\right],\nn&& \\
         \widehat\var(\hat\delta_{st,\text{HT}}) &\equiv& \sum_{c=1}^b\frac{n_c^2}{n^2}\left[ \frac{r}{n_c-1}(\hat\sigma_{cs,\text{HT}}^2 + \hat\sigma_{ct,\text{HT}}^2)\right.\nn &&+ 
         \frac{r^2z_c(r-z_c)}{n_c^3(n_c-r) + n_c^2z_c(r-z_c)}\sum_{k=1}^{n_c}\sum_{\ell\neq k}\left(y_{kcs}y_{\ell cs}T_{kcs}T_{\ell cs}  +y_{kct}y_{\ell ct}T_{kct}T_{\ell ct}\right)\nn&& + \left.\frac{2r^2z_c(r-z_c)}{n_c^4(r-1) - n_c^2z_c(r-z_c)}\sum_{k=1}^{n_c}\sum_{\ell\neq k} y_{kcs}y_{\ell ct}T_{kcs}T_{\ell ct}\right].
        \end{eqnarray}
       
       We now show that these estimators are conservative (in expectation).  
       First, we begin with the following lemma:
       \begin{lemma}
       Under balanced block randomization, for any treatments $s$ and $t$ with $s \neq t$:
         \begin{eqnarray}
         &&  \E\left(\frac{2r^2z_c(r-z_c)}{n_c^4(r-1) - n_c^2z(r-z)}\sum_{k=1}^{n_c}\sum_{\ell\neq k} y_{kcs}y_{\ell ct}T_{kcs}T_{\ell ct}\right) \nn&&= \frac{2z_c(r-z_c)}{n_c^3(n_c-1)(r-1)}\sum_{k=1}^{n_c}\sum_{\ell\neq k}y_{kcs}y_{\ell ct}.
         \end{eqnarray}
         \label{lemmacrosssum}
       \end{lemma} 
       This lemma is proved in Appendix~\ref{proofblocklevelvarests}.
       
       Also note, by the Cauchy-Schwarz and the AM-GM inequalities respectively, that:
       \begin{eqnarray}
         \gamma_{cst} \leq \sqrt{ \sigma_{cs}^2\sigma_{ct}^2} \leq \frac{\sigma_{cs}^2 + \sigma_{ct}^2}{2}.
         \label{covboundhere}
       \end{eqnarray}
       The first two terms are equal if and only if there exists constants $a$ and $b$ such that, 
       for all $k \in \{1,\ldots, n_c\}$, $y_{kcs} = a + by_{kct}$.  
       The last two terms are equal if and only if $\sigma_{cs}^2 = \sigma_{ct}^2$.
       Hence, \eqref{covboundhere} is satisfied with equality if and only if there exists a constant $a$ such
       that, for all $k \in \{1,\ldots, n_c\}$, $y_{kcs} = a + y_{kct}$; that is, if and only if 
       treatment shifts the value of the potential outcomes by a constant for all units within block $c$.
       
       \begin{theorem}
         Under balanced block randomization, for any treatments $s$ and $t$ with $s \neq t$:
         \begin{eqnarray}
           \E( \widehat\var(\hat\delta_{st,\text{diff}})  ) \geq \var(\hat\delta_{st,\text{diff}}),~~~~ 
           \E( \widehat\var(\hat\delta_{st,\text{HT}})  ) \geq \var(\hat\delta_{st,\text{HT}}).
           \label{blockvarianceareunbiased}
	\end{eqnarray}
	with equality if and only if, for each block $c$, there is a constant $a_c$ such that,
	for all $k \in \{1, \ldots, n_c\}$, $y_{kcs} = a_c + y_{kct}$.
       $\sigma_{cs}^2 = \sigma_{ct}^2$.
       \end{theorem}
       
       \begin{proof}
       Define:
       \begin{eqnarray}
         \var^*_{st,\text{diff}} &\equiv& \sum_{c=1}^b \frac{n^2_c}{n^2}
           \left(\frac{r}{n_c-1}(\sigma_{cs}^2+\sigma_{ct}^2)\right)
         \nn&&+ \sum_{c=1}^b \frac{n^2_c}{n^2}
           \left(\frac{rz_c(r-z_c)}{(n_c-1)(n_c-z_c)(n_c+r-z_c)}(\sigma^2_{cs} + \sigma^2_{ct})\right) \\
            \var^*_{st,\text{HT}} &\equiv& \sum_{c=1}^b\frac{n^2_c}{n^2}
            \left( \frac{r}{n_c-1}(\sigma_{cs}^2+\sigma_{ct}^2) \right)\nn&&
            +  \sum_{c=1}^b\frac{z_c(r-z_c)}{n_c^3(r-1)}\sum_{k=1}^{n_c}\sum_{\ell \neq k}\left(
              \frac{r-1}{n_c -1}\left( y_{kcs}y_{\ell cs} + y_{kct}y_{\ell ct}\right) + 2\frac{y_{kcs}y_{\ell ct}}{n_c-1}\right)\nn
       \end{eqnarray}
       By Theorems~\ref{brdiffthm} and~\ref{brhtthm}, and by equation~\eqref{covboundhere}, it follows that:
       \begin{equation}
         \var^*_{st,\text{diff}} \geq  \var(\hat\delta_{st,\text{diff}}),~~~~ 
         \var^*_{st,\text{HT}} \geq \var(\hat\delta_{st,\text{HT}}),
         \label{boundvarstar}
       \end{equation}
       with equality if and only if, for each block $c$, there is a constant $a_c$ such that,
	for all $k \in \{1, \ldots, n_c\}$, $y_{kcs} = a_c + y_{kct}$.
       Moreover, by Lemmas~\ref{blocklevelvarests} and~\ref{lemmacrosssum}, 
       and by linearity of expectations,
       we have that:
       \begin{equation}
          \E( \widehat\var(\hat\delta_{st,\text{diff}})  ) =  \var^*_{st,\text{diff}},~~~~
            \E( \widehat\var(\hat\delta_{st,\text{HT}})  )  =  \var^*_{st,\text{HT}}.
       \end{equation}
       The theorem immediately follows.
      \end{proof} 
      
  \pdfbookmark[1]{References}{sec:references}

\newpage

\section*{Appendices} 
    \appendix
    \section{Proof of Lemmas~\ref{lemmadiffblockests} and~\ref{lemmahtblockests} 
      \label{proofoflemmasforvarcals}}
    The following proofs use methods found in~\citet{cochran1977sampling}
    and~\citet{lohr1999sampling}.
    Additionally, the variance calculations for the sample mean 
    follow~\citet{miratrix2012adjusting} closely.
    To help the reader, we refer each unit by a single index.
    
   For any distinct units $i$ and $j$, and distinct treatments $s$ and $t$, 
    the following expectations hold under complete randomization:
     \begin{eqnarray}
      \E\left(\frac{T_{is}}{\#T_{s}}\right)&=& 
        \E\left[\E\left(\left.\frac{T_{is}}{\#T_{s}}\right| \#T_{s}\right)\right] \nn&=& 
        \E\left(\frac{\frac{\#T_{s}}{n}}{\#T_{s}}\right) = \E\left(\frac{1}{n}\right) = \frac 1 n,\label{exp1samp}\\
      \E\left(\frac{T_{is}}{(\#T_{s})^2}\right) &=& 
        \E\left[\E\left(\left.\frac{T_{is}}{(\#T_{s})^2}\right| \#T_{is}\right)\right] \nn&=& 
        \E\left(\frac{\frac{\#T_{s}}{n}}{(\#T_{s})^2}\right) = \E\left(\frac{1}{n}\frac{1}{\#T_{s}}\right) 
        = \frac 1 n\E\left(\frac{1}{\#T_{s}}\right),\label{exp2samp}\\
      \E\left(\frac{T_{is}T_{js}}{(\#T_{s})^2}\right) &=& 
        \E\left[\E\left(\left.\frac{T_{is}T_{js}}{(\#T_{s})^2}\right| \#T_{is}\right)\right] \nn&=& 
        \E\left(\frac{\frac{\#T_{s}}{n}\frac{\#T_{s}-1}{n-1}}{(\#T_{s})^2}\right)  
          = \E\left(\frac{(\#T_{s})^2 - \#T_{s}}{n(n-1)(\#T_{s})^2}\right)  \nn &=&
          \frac{1}{n(n-1)}\E\left(1-\frac 1{\#T_s}\right) \nn &=&
         \frac{1}{n(n-1)} -   \frac{1}{n(n-1)}\E\left(\frac{1}{\#T_{s}}\right),\label{exp3samp}\\
      \E\left(\frac{T_{is}T_{jt}}{\#T_{s}\#T_{t}}\right) &=& 
        \E\left[\E\left(\left.\frac{T_{is}T_{jt}}{\#T_{s}\#T_{t}}\right| \#T_{s},\#T_{t}\right)\right] \nn&=& 
        \E\left(\frac{\frac{\#T_{s}}{n} \frac{\#T_{t}}{n-1}}{\#T_{s}\#T_{t}}\right) 
        = \E\left(\frac{1}{n(n-1)}\right) = \frac 1 {n(n-1)}.\label{exp4samp}   
    \end{eqnarray}
    
    We first compute the expectation of the block-level estimator
    $\hat{\mu}_{s,\text{samp}}$.  
    By~\ref{exp1samp},
    \begin{equation}
      \E(\hat{\mu}_{s,\text{samp}}) = \E\left(\sum_{i=1}^n \frac{y_{is}T_{is}}{\#T_{s}}\right)
      = \sum_{i=1}^n y_{is}\E\left(\frac{T_{is}}{\#T_{s}}\right)
      =  \sum_{i=1}^n \frac{y_{is}}{n} = \mu_s.
      \label{corrmean}
    \end{equation}
    
    We now derive the variance of this estimator.     
    Observe that, by~\eqref{exp2samp} and~\eqref{exp3samp}:
    \begin{eqnarray}
      &&\E\left(\hat{\mu}_{s,\text{diff}}^2\right)  =  
        \E\left[\left(\sum_{i=1}^n \frac{y_{is}T_{is}}{\#T_{s}}\right)^2\right]\nn &=&
        \E\left[\sum_{i=1}^n \frac{y^2_{is}T^2_{is}}{(\#T_{s})^2}
          + \sum_{i=1}^n\sum_{j\neq i} \frac{y_{is}y_{js}T_{is}T_{js}}{(\#T_{s})^2}\right] =
        \E\left[\sum_{i=1}^n \frac{y^2_{is}T_{is}}{(\#T_{s})^2}
          + \sum_{i=1}^n\sum_{j\neq i} \frac{y_{is}y_{js}T_{is}T_{js}}{(\#T_{s})^2}\right] \nn &=&
        \sum_{i=1}^n y^2_{is}\E\left(\frac{T_{is}}{(\#T_{s})^2}\right)
          + \sum_{i=1}^n\sum_{j\neq i} y_{is}y_{js}\E\left(\frac{T_{is}T_{js}}{(\#T_{s})^2}\right)\nn&=&
        \frac 1 n\E\left(\frac 1 {\#T_s}\right)\sum_{i=1}^n y^2_{is}
          + \sum_{i=1}^n\sum_{j\neq i} y_{is}y_{js}\left(
          \frac 1 {n(n-1)}  - \frac 1 {n(n-1)}\E\left(\frac{1}{\#T_{s}}\right)\right)\nn&=&
        \frac 1 n\E\left(\frac 1 {\#T_s}\right)\sum_{i=1}^n y^2_{is}
          + \left(\left(\sum_{i=1}^ny_{is}\right)^2 - \sum_{i=1}^n y_{is}^2\right)\left(
          \frac 1 {n(n-1)}  - \frac 1 {n(n-1)}\E\left(\frac{1}{\#T_{s}}\right)\right) \nn &=&
        \frac 1{n(n-1)}\left(\sum_{i=1}^ny_{is}\right)^2 - \frac 1{n(n-1)}\sum_{i=1}^n y_{is}^2 \nn &&
          + \E\left(\frac 1 {\#T_s}\right)\left(\left(\frac 1 n + \frac 1 {n(n-1)}\right)\sum_{i=1}^n y_{is}^2
          -  \frac 1{n(n-1)}\left(\sum_{i=1}^ny_{is}\right)^2\right). \label{briefstopsampvariance}
        \end{eqnarray}
        We can simplify the last term in parentheses.  
        \begin{eqnarray}
          &&\left(\frac 1 n + \frac 1 {n(n-1)}\right)\sum_{i=1}^n y_{is}^2
          -  \frac 1{n(n-1)}\left(\sum_{i=1}^ny_{is}\right)^2\nn&=&
          \frac 1 {n-1} \sum_{i=1}^n y_{is}^2
          -  \frac 1{n(n-1)}\left(\sum_{i=1}^ny_{is}\right)^2 \nn & = &
          \frac 1 {n-1} \sum_{i=1}^n y_{is}^2
          -  \frac n{n-1}\left(\sum_{i=1}^n\frac{y_{is}}{n}\right)^2 \nn &=&
					 \frac n {n-1} \sum_{i=1}^n \frac{y_{is}^2}{n}
          -  \frac n{n-1}\left(\sum_{i=1}^n\frac{y_{is}}{n}\right)^2 \nn &=&
					 \frac n {n-1} \left(\sum_{i=1}^n \frac{y_{is}^2}{n}
          -  \left(\sum_{i=1}^n\frac{y_{is}}{n}\right)^2\right) = \frac{n}{n-1}\sigma_{s}^2.
					\label{simplifiedparenstermsampvar}
        \end{eqnarray}
				The last equality is obtained by applying~\eqref{definedomainvar}.
      
			Continuing from~\eqref{briefstopsampvariance} and applying~\eqref{simplifiedparenstermsampvar}, we find that:
			\begin{eqnarray}
			  \E\left(\hat{\mu}_{s,\text{diff}}^2\right) &=&
				 \frac 1{n(n-1)}\left(\sum_{i=1}^ny_{is}\right)^2 - \frac 1{n(n-1)}\sum_{i=1}^n y_{is}^2 \nn &&
          + \E\left(\frac 1 {\#T_s}\right)\left(\left(\frac 1 n + \frac 1 {n(n-1)}\right)\sum_{i=1}^n y_{is}^2
          -  \frac 1{n(n-1)}\left(\sum_{i=1}^ny_{is}\right)^2\right) \nn &=&
					\frac n{n-1}\left(\sum_{i=1}^n\frac{y_{is}}{n}\right)^2 - \frac 1{n-1}\sum_{i=1}^n \frac{y_{is}^2}{n}
					+\E\left(\frac 1 {\#T_s}\right)\frac{n}{n-1}\sigma_{s}^2.
					\label{expsquaredsamplevariance}
			\end{eqnarray}
    
    Since $\var(X) = \E(X^2) - (\E(X))^2$, it follows from~\eqref{expsquaredsamplevariance} and~\eqref{definedomainvar} that:
    \begin{eqnarray}
      \var(\hat{\mu}_{s,\text{diff}}) &=& \E(\hat{\mu}_s^2) - (\E(\hat{\mu}_s))^2\nn & = &
       \frac n{n-1}\left(\sum_{i=1}^n\frac{y_{is}}{n}\right)^2 - \frac 1{n-1}\sum_{i=1}^n \frac{y_{is}^2}{n}
					+\E\left(\frac 1 {\#T_s}\right)\frac{n}{n-1}\sigma_{s}^2 - \left(\sum_{i=1}^n\frac{y_{is}}{n}\right)^2\nn &=&
          \frac 1{n-1}\left(\sum_{i=1}^n\frac{y_{is}}{n}\right)^2 - \frac 1{n-1}\sum_{i=1}^n \frac{y_{is}^2}{n}
					+\E\left(\frac 1 {\#T_s}\right)\frac{n}{n-1}\sigma_{s}^2  \nn &=&
					\frac {-1}{n-1}\left(\frac 1{n-1}\sum_{i=1}^n \frac{y_{is}^2}{n} - \left(\sum_{i=1}^n\frac{y_{is}}{n}\right)^2\right)
					+\E\left(\frac 1 {\#T_s}\right)\frac{n}{n-1}\sigma_{s}^2  \nn &=&
          \frac {-1}{n-1}\sigma_s^2 +\E\left(\frac 1 {\#T_s}\right)\frac{n}{n-1}\sigma_{s}^2 
					= \frac{n}{n-1}\left(\E\left(\frac 1 {\#T_s}\right) - \frac{1}{n}\right)\sigma_s^2.
           \label{samplevarcalceasy}
    \end{eqnarray}
    
    Note that $\lfloor n/r \rfloor = n/r - z/r$. 
    Thus, under complete randomization:
    \begin{eqnarray}
      \E\left(\frac{1}{\#T_s}\right) &=& \frac z r\left(\frac 1 {\lfloor n/r \rfloor + 1}\right)
           + \left(1-\frac z r\right)\left(\frac 1 {\lfloor n/r \rfloor}\right) \nn &=&
         \frac{z\lfloor n/r \rfloor}{r(\lfloor n/r \rfloor)(\lfloor n/r \rfloor + 1)}
           + \frac{(r - z)(\lfloor n/r \rfloor + 1)}{r(\lfloor n/r \rfloor)(\lfloor n/r \rfloor + 1)} \nn&=&
			   \frac{z(n/r - z/r) + (r - z)(n/r - z/r + 1)}{r(n/r - z/r)(n/r - z/r + 1)}
       \nn&=&
					\frac{(1/r)z(n - z) + (1/r)(r - z)(n - z + r)}{(1/r)(n - z)(n +r - z)} 
				\nn&=& \frac{z(n-z) + (r-z)(n - z + r)}{(n-z)(n+r-z)} \nn &=&
		\frac{r(n-z) + r^2 -zr}{(n-z)(n+r-z)} = \frac{nr + r^2 - 2rz}{(n-z)(n+r-z)}.
    \end{eqnarray}
    It follows that:
    \begin{eqnarray}
      \var(\hat{\mu}_{s,\text{diff}}) &=& \frac n {n-1}\sigma_s^2
          \left(\E\left(\frac 1 {\#T_s}\right) - \frac 1 n\right) \nn&=&
        \frac n{n-1}\sigma_s^2
          \left(\frac{nr + r^2 - 2rz}{(n-z)(n+r - z)} - \frac 1 n\right) \nn &=&
         \frac n{n-1}\frac{n^2r+nr^2 - 2nrz - (n-z)(n+r - z)}{n(n-z)(n+r-z)}\sigma_s^2\nn&=&
         \frac{nr(n+r - z)  - nrz - (n-z)(n+r - z)}{(n-1)(n-z)(n+r-z)}\sigma_s^2\nn&=&
				\frac{nr(n+r - z) - rz(n + r - z) - (n-z)(n+r - z)+ rz(r-z)}{(n-1)(n-z)(n+r-z)}\sigma_s^2\nn&=&
				\frac{(nr- rz -n + z)(n+r - z)+ rz(r-z)}{(n-1)(n-z)(n+r-z)}\sigma_s^2\nn&=&
				\frac{(r-1)(n-z)(n+r - z)+ rz(r-z)}{(n-1)(n-z)(n+r-z)}\sigma_s^2\nn&=&
				\left(\frac{r-1}{n-1} + \frac{rz(r-z)}{(n-1)(n-z)(n+r-z)}\right)\sigma_s^2.
    \end{eqnarray}
    We now derive covariances of this estimator. 
    Note that:
    \begin{eqnarray}
      && \E\left(\sum_{i=1}^n \frac{y_{is}T_{is}}
          {\#T_{s}}\sum_{i=1}^n\frac{y_{it}T_{it}}{\#T_{t}}\right)\nn &=&
         \E\left(\sum_{i=1}^n\sum_{j\neq i} \frac{y_{is}y_{jt}T_{is}T_{jt}}{\#T_{s}\#T_{t}}\right)
           +  \E\left(\sum_{i=1}^n\frac{y_{is}y_{it}T_{is}T_{it}}{\#T_{s}\#T_{t}}\right) \nn &=&
         \sum_{i=1}^n\sum_{j\neq i}y_{is}y_{jt} \E\left(\frac{T_{is}T_{jt}}{\#T_{s}\#T_{t}}\right)
           + \sum_{i=1}^ny_{is}y_{it}\E\left(\frac{ T_{is}T_{it}}{\#T_{s}\#T_{t}}\right) \nn&=&
         \sum_{i=1}^n\sum_{j\neq i}y_{is}y_{jt} \frac 1 {n(n-1)}+ 0 =
           \sum_{i=1}^n\sum_{j\neq i}\frac{y_{is}y_{jt}}{n(n-1)}.	
    \end{eqnarray}

    Recall $\cov(X,Y) = \E(XY) -\E(X)\E(Y)$.
    It follows that:
    \begin{eqnarray}
       \cov(\hat \mu_{s,\text{diff}},\hat \mu_{t,\text{diff}}) &=&
         \E\left(\sum_{i=1}^n \frac{y_{is}T_{is}}{\#T_{is}}\sum_{i=1}^n\frac{y_{it}T_{it}}{\#T_{it}}\right)
           - \sum_{i=1}^n \frac{y_{is}}{n}\sum_{i=1}^n \frac{y_{it}}{n} \nn &=&
         \sum_{i=1}^n\sum_{j\neq i}\frac{y_{is}y_{jt}}{n(n-1)} - 
           \sum_{i=1}^n \frac{y_{is}}{n}\sum_{i=1}^n \frac{y_{it}}{n} \nn &=&
         \frac{1}{n(n-1)}\sum_{i=1}^n y_{is}\sum_{i=1}^n y_{it} -
           \frac{1}{n(n-1)}\sum_{i=1}^n y_{is}y_{it} -
           \frac{1}{n^2}\sum_{i=1}^n y_{is}\sum_{i=1}^n y_{it}\nn&=&
         \frac{1}{n^2(n-1)}\sum_{i=1}^n y_{is}\sum_{i=1}^n y_{it} 
         -  \frac{1}{n(n-1)}\sum_{i=1}^n y_{is}y_{it} \nn &=&
         \frac{-1}{n-1}\left(\sum_{i=1}^n \frac{y_{is}y_{it}}{n} - 
           \sum_{i=1}^n \frac{y_{is}}n\sum_{i=1}^n\frac{ y_{it}}n \right) = \frac{-\gamma_{st}}{n-1}.
    \end{eqnarray}
    Our derivation of the variance and covariance
		of the sample mean show that 
		the variance and covariate expressions
		derived in~\citet{miratrix2012adjusting}
		are incorrect by a factor of $\frac{n}{n-1}$.

  We now turn our attention to the Horvitz-Thompson estimator.
 For any distinct units $i$ and $j$, and distinct treatments $s$ and $t$, 
    the following expectations hold under complete randomization:
  \begin{eqnarray}
	  \E\left(\#T_s\right) & = & \E\left(\sum_{i=1}^nT_{is}\right)
	  = \sum_{i=1}^n \E(T_{is}) = \sum_{i=1}^n 1/r = n/r, \label{expHT1} \\
    \E\left((\#T_s)^2\right) & = &  \frac z r \left(\lfloor n/r \rfloor + 1\right)^2
           + \left(1 - \frac z r\right)\left(\lfloor n/r \rfloor\right)^2  \nn &=&
         \frac z r\left((\lfloor n/r \rfloor)^2 + 2\lfloor n/r \rfloor + 1\right)
           + \left(1 - \frac z r\right)(\lfloor n/r \rfloor)^2 \nn &=&
         (\lfloor n/r \rfloor)^2 + \frac {2z} r\lfloor n/r \rfloor + \frac z r 
           =(\lfloor n/r \rfloor + z/r)^2 + (z/r - (z/r)^2) \nn &=&
         (n/r)^2 + z/r(1-z/r),\label{expHT2}
 \end{eqnarray}
\begin{eqnarray}    
    \E\left(\#T_s\#T_t\right) & = &  \frac{z(z-1)}{r(r-1)}\left(\lfloor n/r \rfloor + 1\right)^2 + 
         \frac{(r-z)(r-z-1)}{r(r-1)}\left(\lfloor n/r \rfloor\right)^2\nn &&+
         \frac{2z(r-z)}{r(r-1)}\left(\lfloor n/r \rfloor + 1\right)\left(\lfloor n/r \rfloor\right)\nn &=&
         \frac 1 {r(r-1)}\left(
           \begin{array}{l}
           z(z-1)((\lfloor n/r\rfloor)^2 + 2\lfloor n/r \rfloor + 1) \\
           +(r-z)(r-z-1)(\lfloor n/r \rfloor)^2\\
           +2z(r-z)((\lfloor n/r\rfloor)^2 + \lfloor n/r\rfloor) 
           \end{array}
         \right)\nn&=&
         \frac 1 {r(r-1)}\left(
           \begin{array}{l}
           (z(z-1)+(r-z)(r-z-1) + 2z(r-z))(\lfloor n/r\rfloor)^2\\
            +(2z(z-1) +2z(r-z))\lfloor n/r \rfloor\\
            +z(z-1)
           \end{array}
         \right)\nn&=&
          \frac 1 {r(r-1)}\left(
           \begin{array}{l}
           (z(z-1)+(r-z)(r+z-1))(\lfloor n/r\rfloor)^2\\
            +(2z(r-1))\lfloor n/r \rfloor
            +z(z-1)
           \end{array}\right) \nn&=&
           \frac 1 {r(r-1)}\left(
           \begin{array}{l}
           (z^2 - z + r^2 -z^2 -r  +z)(n/r - z/r)^2\\
            +(2z(r-1))(n/r - z/r)
            +z(z-1)
           \end{array}\right) \nn&=&
            \frac 1 {r^3(r-1)}\left(
           (r^2 -r )(n - z)^2
            +r(r-1)2z(n - z)
            +r^2z(z-1)
         \right) \nn&=&
            \frac 1 {r^3(r-1)}\left(r(r-1)\left((n-z)^2 + 2z(n-z)\right) + r^2z(z-1)\right)\nn&=&
             \frac 1 {r^2(r-1)}\left((r-1)\left((n-z)^2 + 2z(n-z) + z^2\right) -z^2(r-1) + rz(z-1)\right)\nn&=&
             \frac 1 {r^2(r-1)}\left((r-1)(n-z+z)^2 -rz +z^2\right)\nn&=&
             \frac {n^2(r-1) - z(r-z)} {r^2(r-1)}.\label{expHT3}
         \end{eqnarray}
         
         \pagebreak
         Using these expressions, we can compute the following expectations under
         complete randomization, assuming distinct treatments $s$ and $t$ and 
         distinct units $i$ and $j$:
         \begin{eqnarray}
        \E(T_{is}T_{js}) &=&  =  \E\left(\E(T_{is}T_{js} | \#T_s )\right) = 
         \E\left(\frac{\#T_s(\#T_s-1)}{n(n-1)}\right) \nn&=&
         \frac{\E[(\#T_s)^2] - \E[\#T_s]}{n(n-1)} =
          \frac{ (n/r)^2 + z/r(1-z/r) - n/r}{n(n-1)} \nn &=&
        \frac{ (n/r)^2 - (z/r)^2 - (n/r - z/r)}{n(n-1)} \nn &=& 
        \frac{ n^2 - z^2 - (nr - zr)}{n(n-1)r^2} \nn &=& 
        \frac{ (n - z)(n + z) - (nr - zr)}{n(n-1)r^2} \nn&=&
        \frac{ (n - z)(n  -r + z)}{n(n-1r^2)}  = \frac{n(n-r) + z(r-z)}{n(n-1)r^2}, 
          \label{jointindicatexp1} \\
        \E(T_{is}T_{jt}) & = & \E[\E(T_{is}T_{jt}| \#T_s,\#T_t)] \nn &=&
         \E\left(\frac{\#T_s\#T_t}{n(n-1)}\right)  =
        \frac{ \E(\#T_s\#T_t)}{n(n-1)} \nn&=&
        \frac{n^2(r-1) - z(r-z)}{n(n-1)r^2(r-1)}.  \label{jointindicateexp2}
    \end{eqnarray}
    
  Under
  complete randomization, the expectation of the Horvitz-Thompson estimator is:
  \begin{equation}
    \E(\hat\mu_{s,\text{HT}}) = \E\left(\sum_{i=1}^n \frac{y_{is}T_{is}}{n/r}\right)
      =\sum_{i=1}^n \frac{y_{is}\E(T_{is})}{n/r} 
      =\sum_{i=1}^n \frac{y_{is}(1/r)}{n/r}
      =\sum_{i=1}^n \frac{y_{is}}{n} = \mu_s.
  \end{equation}
  
  The variance of this estimator is derived as follows.
  By~\eqref{jointindicatexp1}:
  \begin{eqnarray}
   \E(\hat\mu_{s,\text{HT}}^2) &=& \E\left(\left(\sum_{i=1}^n \frac{y_{is}T_{is}}{n/r}\right)^2\right)
          \nn&=&
        \left(\frac {r^2} {n^2}\right)\left(\E\left(\sum_{i=1}^n y^2_{is}T^2_{is}\right)
          + \E\left(\sum_{i=1}^n\sum_{j\neq i} y_{is}y_{js}T_{is}T_{js} \right)\right)\nn &=&
        \left(\frac {r^2} {n^2}\right)\left(\E\left(\sum_{i=1}^n y^2_{is}T_{is}\right)
          + \E\left(\sum_{i=1}^n\sum_{j\neq i} y_{is}y_{js}T_{is}T_{js} \right)\right)\nn &=&
        \left(\frac {r^2} {n^2}\right)\left(\sum_{i=1}^n y^2_{is}\E(T_{is})
          + \sum_{i=1}^n\sum_{j\neq i} y_{is}y_{js}\E(T_{is}T_{js}) \right) \nn &=&
        \left(\frac {r^2}{n^2}\right)\left(\sum_{i=1}^n \frac{y^2_{is}}{r}
          + \sum_{i=1}^n\sum_{j\neq i} y_{is}y_{js}
           \frac{ n(n-r) + z(r-z)}{r^2n(n-1)} \right)\nn &=&
        \left(\frac {r^2}{n^2}\right)\left[\sum_{i=1}^n \frac{y^2_{is}}{r}
          + \left( \frac{ n(n-r) + z(r-z)}{r^2n(n-1)}  \right)
          \left( \left(\sum_{i=1}^n y_{is}\right)^2 - \sum_{i=1}^n y_{is}^2 
          \right)\right]\nn &=&
        \left(\frac{r}{n^2} - \frac{n(n-r) + z(r-z)}{n^3(n-1)}\right)\sum_{i=1}^n y_{is}^2 \nn&&
          + \left(\frac{ n(n-r) + z(r-z)}{n^3(n-1)}\right)\left(\sum_{i=1}^n y_{is}\right)^2  \nn &=&
        \left(\frac{rn(n-1)- n(n-r) - z(r-z)}{n^2(n-1)}\right)\sum_{i=1}^n \frac{y_{is}^2}{n} \nn&&
          + \left(\frac{ n(n-r) + z(r-z)}{n(n-1)}\right)\left(\sum_{i=1}^n \frac{y_{is}}n\right)^2 \nn&=&
        \left(\frac{n^2r -n^2 -z(r-z)}{n^2(n-1)}\right)\sum_{i=1}^n \frac{y_{is}^2}{n} 
          + \left(\frac{ n(n-r) + z(r-z)}{n(n-1)}\right)\left(\sum_{i=1}^n \frac{y_{is}}n\right)^2    \nn &=&
           \left(\frac{n^2(r-1) -z(r-z)}{n^2(n-1)}\right)\sum_{i=1}^n \frac{y_{is}^2}{n} 
          + \left(\frac{ n(n-r) + z(r-z)}{n(n-1)}\right)\left(\sum_{i=1}^n \frac{y_{is}}n\right)^2.  \nn
          \label{goodbreakpointht}
    \end{eqnarray}

    It follows by~\eqref{goodbreakpointht} and~\eqref{definedomainvar} that:
    \begin{eqnarray}
      && \var(\hat \mu_{s,\text{HT}}) = \E(\hat\mu_{s,\text{HT}}^2) - (\E(\hat\mu_{s,\text{HT}}))^2 \nn &=&
        \left(\frac{n^2(r-1) -z(r-z)}{n^2(n-1)}\right)\sum_{i=1}^n \frac{y_{is}^2}{n} 
           \nn&& + \left(\frac{n(n-r) + z(r-z)}{n(n-1)}\right)\left(\sum_{i=1}^n \frac{y_{is}}n\right)^2
          - \left(\sum_{i=1}^n \frac{y_{is}}n\right)^2 \nn&=&
       \left(\frac{n^2(r-1) -z(r-z)}{n^2(n-1)}\right)\sum_{i=1}^n \frac{y_{is}^2}{n} \nn&&
          + \left(\frac{n(n-r) + z(r-z)- n(n-1)}{n(n-1)}\right)
          \left(\sum_{i=1}^n \frac{y_{is}}n\right)^2 \nn&=&
          \left(\frac{n^2(r-1) -z(r-z)}{n^2(n-1)}\right)\sum_{i=1}^n \frac{y_{is}^2}{n}+ \left(\frac{n(1-r) + z(r-z)}{n(n-1)}\right)
          \left(\sum_{i=1}^n \frac{y_{is}}n\right)^2 \nn&=&
        \left(\frac{n^2(r-1) -z(r-z)}{n^2(n-1)}\right)\sum_{i=1}^n \frac{y_{is}^2}{n} 
                  - \left(\frac{n(r-1) - z(r-z)}{n(n-1)}\right)
          \left(\sum_{i=1}^n \frac{y_{is}}n\right)^2 \nn&=&
          \frac{r-1}{n-1}\left(\sum_{i=1}^n \frac{y_{is}^2}{n} - \left(\sum_{i=1}^n \frac{y_{is}}n\right)^2\right)  - \frac{z(r-z)}{n^3(n-1)}\left( \sum_{i=1}^n y^2_{is} - \left(\sum_{i=1}^n y_{is}\right)^2\right) \nn&=&
          \frac{r-1}{n-1}\sigma_s^2  - \frac{z(r-z)}{n^3(n-1)}\left( -\sum_{i=1}^n\sum_{j\neq i} y_{is}y_{js}\right)            \nn&=&             
           \frac{r-1}{n-1}\sigma_s^2  + \frac{z(r-z)}{n^3(n-1)}\sum_{i=1}^n\sum_{j\neq i} y_{is}y_{js}.
         \label{decomposedfurtherht}
    \end{eqnarray}
  
    Now we derive the covariance.
    Note that:
    \begin{eqnarray}
      && \E\left(\sum_{i=1}^n \frac{y_{is}T_{is}}{n/r}\sum_{i=1}^n\frac{y_{it}T_{it}}{n/r}\right)\nn&=& 
        \E\left(\sum_{i=1}^n\sum_{j\neq i} \frac{y_{is}y_{jt}T_{is}T_{jt}}{(n/r)^2}\right)
          + \E\left(\sum_{i=1}^n\frac{y_{is}y_{it}T_{is}T_{it}}{(n/r)^2}\right) \nn &=&
        \sum_{i=1}^n\sum_{j\neq i} \frac{y_{is}y_{jt}}{(n/r)^2}\E\left(T_{is}T_{jt}\right)
          + \sum_{i=1}^n\frac{y_{is}y_{it}}{(n/r)^2}\E\left(T_{is}T_{it}\right) \nn &=&
         \sum_{i=1}^n\sum_{j\neq i} \frac{y_{is}y_{jt}}{(n/r)^2}
           \frac{n^2(r-1) - z(r-z)}{n(n-1)r^2(r-1)} + 0\nn &=&
         \sum_{i=1}^n\sum_{j\neq i} \frac{y_{is}y_{jt}(n^2(r-1) - z(r-z))}{n^3(n-1)(r-1)} 
          \nn&=&
         \sum_{i=1}^n\sum_{j\neq i} 
           \frac{y_{is}y_{jt}}{n(n-1)} - \sum_{i=1}^n\sum_{j\neq i}\frac{y_{is}y_{jt}z(r-z)}{n^3(n-1)(r-1)}.\label{sumexpecteasysubform} 
    \end{eqnarray}
  
  Thus, using $\cov(X,Y) = \E(XY) - \E(X)\E(Y)$ and 
  applying~\eqref{sumexpecteasysubform} and~\eqref{definedomaincov}, we have:
     \begin{eqnarray}
       \cov(\hat \mu_{s,\text{HT}},\hat \mu_{t,\text{HT}}) &=& 
       \E(\hat \mu_{s,\text{HT}}\hat \mu_{t,\text{HT}}) - \E(\hat \mu_{s,\text{HT}})\E(\hat \mu_{t,\text{HT}})\nn&=&
         \E\left(\sum_{i=1}^n \frac{y_{is}T_{is}}{n/r}\sum_{i=1}^n\frac{y_{it}T_{it}}{n/r}\right)
           - \sum_{i=1}^n \frac{y_{is}}{n}\sum_{i=1}^n \frac{y_{it}}{n} \nn &=&
         \sum_{i=1}^n\sum_{j\neq i}\frac{y_{is}y_{jt}}{n(n-1)} 
           - \sum_{i=1}^n\sum_{j\neq i}\frac{y_{is}y_{jt}z(r-z)}{n^3(n-1)(r-1)} - 
           \sum_{i=1}^n \frac{y_{is}}{n}\sum_{i=1}^n \frac{y_{it}}{n} \nn &=&
         \frac{-\gamma_{st}}{n-1} - \sum_{i=1}^n\sum_{j\neq i}\frac{y_{is}y_{jt}z(r-z)}{n^3(n-1)(r-1)}.
    \end{eqnarray}

  This proves the two lemmas.
  \section{Proof of Lemmas~\ref{blocklevelvarests} and~\ref{lemmacrosssum} 
    \label{proofblocklevelvarests}}
  To help the reader, we suppress the block index in the following derivations,
   identifying units by a single index.  
   
   Note that, under complete randomization and for distinct treatments $s$ and $t$
   and distinct units $i$ and $j$,
   the following expectations hold:
  \begin{eqnarray}
    \E\left(\frac{T_{is}}{\#T_{s} - 1}\right) &=& 
        \E\left[\E\left(\left.\frac{T_{is}}{\#T_{s}-1}\right| \#T_{s}\right)\right] \nn&=& 
        \E\left(\frac{\frac{\#T_{s}}{n}}{\#T_{s} -1}\right)= 
        \frac 1 n\E\left(\frac{\#T_{s}}{\#T_{s}-1}\right)\nn& = &
        \frac 1 n\E\left(\frac{\#T_{s} - 1}{\#T_{s}-1}\right) + \frac 1 n\E\left(\frac{1}{\#T_{s}-1}\right)\nn&=&
        \frac 1 n + \frac 1 n\E\left(\frac{1}{\#T_{s}-1}\right),\label{overnminus1samp1}\\
      \E\left(\frac{T_{is}}{\#T_s(\#T_{s} - 1)}\right) &=& 
        \E\left[\E\left(\left.\frac{T_{is}}{\#T_s(\#T_{s} - 1)}\right| \#T_{s}\right)\right] \nn&=& 
        \E\left(\frac{\frac{\#T_{s}}{n}}{\#T_s(\#T_{s} - 1)}\right)= 
        \frac 1 n\E\left(\frac{\#T_{s}}{\#T_s(\#T_{s} - 1)}\right)\nn& = &
        \frac 1 n\E\left(\frac{1}{\#T_{s}-1}\right),\label{overnminus1samp2}\\
      \E\left(\frac{T_{is}T_{js}}{\#T_{s}(\#T_s - 1)}\right) &=&
        \E\left[\E\left(\left.\frac{T_{is}T_{js}}{\#T_{s}(\#T_s - 1)}\right| \#T_{s}\right)\right] \nn&=& 
        \E\left(\frac{\frac{\#T_{s}}{n}\frac{\#T_s-1}{n-1}}{\#T_{s}(\#T_s - 1)}\right)\nn&=& 
        \frac 1 {n(n-1)}\E\left(\frac{\#T_{s}(\#T_s - 1)}{\#T_{s}(\#T_s - 1)}\right) = \frac 1 {n(n-1)}.
        \label{overnminus1samp3}
   \end{eqnarray}
   We show that $\E(\hat \sigma^2_{s,\text{samp}}) = \sigma^2_s$.
   The fact that $\E\left[\widehat{ \var}(\hat\mu_{s,\text{samp}})\right] = \var(\hat\mu_{s,\text{samp}})$
   follows immediately.
  
	First, note that:
	\begin{eqnarray}
    && \sum_{i=1}^{n}T_{is}\left(y_{is}T_{is}- \sum_{i=1}^{n}\frac{y_{is}T_{is}}{\#T_{s}}\right)^2
    \nonumber\\
      &=&   \sum_{i=1}^{n}T_{is}(y_{is}T_{is})^2- 
      2\sum_{i=1}^{n}T_{is}\left(y_{is}T_{is}\sum_{i=1}^{n}
       \frac{y_{is}T_{is}}{\#T_{s}}\right) + 
       \sum_{i=1}^{n}T_{is}\left(\sum_{i=1}^{n} \frac{y_{is}T_{is}}{\#T_{s}}\right)^2\nonumber\\
      &=& \sum_{i=1}^{n} y^2_{is}T_{is} - 
      2\sum_{i=1}^{n} \left(y_{is}T_{is}\sum_{i=1}^{n}
       \frac{y_{is}T_{is}}{\#T_{s}}\right)
       +  \sum_{i=1}^{n}T_{is}\left(\sum_{i=1}^{n} \frac{y_{is}T_{is}}{\#T_{s}}\right)^2\nonumber\\
        &=&\sum_{i=1}^{n} y^2_{is}T_{is}
          - 2\#T_{s}\left( \sum_{i=1}^{n} \frac{y_{is}T_{is}}{\#T_{s}}\right)^2
          + \#T_{s}\left( \sum_{i=1}^{n} \frac{y_{is}T_{is}}{\#T_{s}}\right)^2\nonumber\\
        &=&\sum_{i=1}^{n} y^2_{is}T_{is}
          - \#T_{s}\left( \sum_{i=1}^{n} \frac{y_{is}T_{is}}{\#T_{s}}\right)^2.
  \end{eqnarray}
	
	Thus,
  \begin{eqnarray}
    \E(\hat \sigma^2_{s,\text{diff}}) &=&\E\left(\frac{n-1}{n}\sum_{i=1}^{n}\frac{T_{is}\left(y_{is} - \sum_{i=1}^n\frac{y_{is}T_{is}}{\#T_{s}}\right)^2}{\#T_{s} -1} \right) \nn&=&
    \frac{n-1}{n}\E\left(\frac{\sum_{i=1}^{n} y^2_{is}T_{is}
          - \#T_{s}\left( \sum_{i=1}^{n} \frac{y_{is}T_{is}}{\#T_{s}}\right)^2}{\#T_s - 1}\right) \nn&=&
          \frac{n-1}{n}\left(\sum_{i=1}^n y^2_{is}\E\left(\frac{T_{is}}{\#T_s - 1}\right) - 
          \E\left( \frac{\left(\sum_{i=1}^n y_{is}T_{is}\right)^2}{\#T_{is}(\#T_{is} - 1)}\right)\right).\nn
  \end{eqnarray}
  By~\eqref{overnminus1samp1}:
  \begin{eqnarray}
     \sum_{i=1}^n y^2_{is}\E\left(\frac{T_{is}}{\#T_s - 1}\right) = \frac{1}{n}\sum_{i=1}^n y_{is}^2
     + \frac{1}{n}\E\left(\frac{1}{\#T_s - 1}\right)\sum_{i=1}^n y_{is}^2.
     \label{simpsampvarstop1}
  \end{eqnarray}
  
  By~\eqref{overnminus1samp2} and~\eqref{overnminus1samp3}:
  \begin{eqnarray}
   && \E\left( \frac{\left(\sum_{i=1}^n y_{is}T_{is}\right)^2}{\#T_{is}(\#T_{is} - 1)}\right) \nn &=&
    \E\left( \frac{\sum_{i=1}^n (y_{is}T_{is})^2}{\#T_{is}(\#T_{is} - 1)}\right) +  \E\left( \frac{\sum_{i=1}^n\sum_{j\neq i} y_{is}y_{js}T_{is}T_{js}}{\#T_{is}(\#T_{is} - 1)}\right) \nn &=&
    \sum_{i=1}^n y^2_{is}\E\left(\frac{T_{is}}{\#T_{is}(\#T_{is} - 1)}\right) + 
    \sum_{j\neq i} y_{is}y_{js}\E\left(\frac{T_{is}T_{js}}{\#T_{is}(\#T_{is} - 1)}\right) \nn &=&
     \frac{1}{n}\E\left(\frac{1}{\#T_s - 1}\right)\sum_{i=1}^n y_{is}^2 + 
    \frac{1}{n(n-1)}\sum_{i=1}^n \sum_{j\neq i} y_{is}y_{js} \nn&=&
    \frac{1}{n}\E\left(\frac{1}{\#T_s - 1}\right)\sum_{i=1}^n y_{is}^2
    + \frac{1}{n(n-1)}\left(\sum_{i=1}^n y_{is}\right)^2 - \frac{1}{n(n-1)}\sum_{i=1}^n y_{is}^2.\nn
    \label{simpsampvarstop2}
  \end{eqnarray}
  
  Thus, by~\eqref{simpsampvarstop1},~\eqref{simpsampvarstop2}, and~\eqref{definedomainvar}, it follows that:
  \begin{eqnarray}
     \E(\hat \sigma^2_{s,\text{diff}}) &=&
     \frac{n-1}{n}\left(\sum_{i=1}^n y^2_{is}\E\left(\frac{T_{is}}{\#T_s - 1}\right) - 
          \E\left( \frac{\left(\sum_{i=1}^n y_{is}T_{is}\right)^2}{\#T_{is}(\#T_{is} - 1)}\right)\right)\nn&=&
        \frac{n-1}{n}\left( \frac{1}{n}\sum_{i=1}^n y_{is}^2
     + \frac{1}{n}\E\left(\frac{1}{\#T_s - 1}\right)\sum_{i=1}^n y_{is}^2\right)\nn&&
     -\frac{n-1}{n}\left(\frac{1}{n}\E\left(\frac{1}{\#T_s - 1}\right)\sum_{i=1}^n y_{is}^2
    - \frac{1}{n(n-1)}\left(\sum_{i=1}^n y_{is}\right)^2 + \frac{1}{n(n-1)}\sum_{i=1}^n y_{is}^2\right)\nn&=&
    \frac{n-1}{n}\left(\frac{1}{n-1}\sum_{i=1}^n y_{is}^2 - \frac{1}{n(n-1)}\left(\sum_{i=1}^n y_{is}\right)^2\right) \nn&=&
    \frac{n-1}{n}\left(\frac{n}{n-1}\sum_{i=1}^n \frac{y_{is}^2}{n} - \frac{n}{n-1}\left(\sum_{i=1}^n \frac{y_{is}}{n}\right)^2\right) \nn&=&
    \frac{n-1}{n}\left(\frac{n}{n-1}\sigma^2_s\right) = \sigma^2_s.
  \end{eqnarray}
  
  We now focus on estimating the variance of Horvitz-Thompson estimators.
  By~\eqref{jointindicatexp1}, 
  for any treatment $s$, the following expectation holds under complete randomization.
  \begin{eqnarray}
      &&\E\left(\frac{n(n-1)r^2}{n(n-r) + z(r-z)} \sum_{i=1}^n\sum_{j\neq i} y_{is}y_{js}T_{is}T_{js}\right) \nn&=&
        \frac{n(n-1)r^2}{n(n-r) + z(r-z)}\sum_{i=1}^n \sum_{j\neq i} y_{is}y_{js}\E(T_{is}T_{js}) \nn&=&
         \frac{n(n-1)r^2}{n(n-r) + z(r-z)} \sum_{i=1}^n\sum_{j\neq i} y_{is}y_{js}\frac{n(n-r) + z(r-z)}{n(n-1)r^2} \nn&=&
          \sum_{i=1}^n\sum_{j\neq i} y_{is}y_{js},\label{crossjointexpectation1}
  \end{eqnarray}
  
  Applying~\eqref{crossjointexpectation1}, we show unbiasedness of the Horvitz-Thompson
  variance estimator under complete randomization:
  \begin{eqnarray}
       \E\left(\hat\sigma^2_{s,\text{HT}}\right) &=& \E\left(\frac{(n-1)r}{n^2}\sum_{i=1}^n y^2_{is}T_{is}\right) \nn&&- \E\left(\frac {(n-1)r^2} {n^2(n-r) + nz(r-z)} \sum_{i=1}^n\sum_{j\neq i} y_{is}y_{js}T_{is}T_{js}\right) \nn &=&
       \frac{(n-1)r}{n^2}\sum_{i=1}^n y^2_{is}\E(T_{is}) \nn&&
         - \frac{1}{n^2}\E\left(\frac{n(n-1)r^2}{n(n-r) + z(r-z)} \sum_{i=1}^n\sum_{j\neq i} y_{is}y_{js}T_{is}T_{js}\right)\nn &=&       
         \frac{(n-1)r}{n^2}\sum_{i=1}^n y^2_{is}(1/r)- \frac{1}{n^2}\sum_{i=1}^n\sum_{j\neq i} y_{is}y_{js} \nn &=&
         \left(\frac{1}{n} -\frac{1}{n^2}\right)\sum_{i=1}^n y^2_{is}- \frac{1}{n^2}\left(\left(\sum_{i=1}^n y_{is}\right)^2 - \sum_{i=1}^n y^2_{is}\right) \nn &=& 
         \frac{1}{n} \sum_{i=1}^n y^2_{is}- \frac{1}{n^2}\left(\sum_{i=1}^n y_{is}\right)^2=
        \sum_{i=1}^n \frac{y^2_{is}}{n}- \left(\sum_{i=1}^n \frac{y_{is}}{n}\right)^2 = \sigma_s^2. \nn
        \label{varhtunbiased}                    
     \end{eqnarray}
     Thus, by~\eqref{jointindicatexp1} and~\eqref{varhtunbiased}, we show unbiasedness of 
     the variance estimator
     $\widehat{ \var}(\hat\mu_{s,\text{HT}})$ under complete randomization:
     \begin{eqnarray}
       \E(\widehat{ \var}(\hat\mu_{s,\text{HT}})) &=&  \E\left(\frac{r-1}{n-1}\hat\sigma_{s,\text{HT}}^2  + \frac{r^2z(r-z)}{n^3(n-r) + n^2z(r-z)}\sum_{i=1}^n\sum_{j\neq i} y_{is}y_{js}T_{is}T_{js}\right)\nn&=& 
       \frac{r-1}{n-1}\E\left(\hat\sigma_{s,\text{HT}}^2\right)
       + \E\left(\frac{r^2z(r-z)}{n^3(n-r) + n^2z(r-z)}\sum_{i=1}^n\sum_{j\neq i} y_{is}y_{js}T_{is}T_{js}\right) \nn &=&
             \frac{r-1}{n-1}\sigma^2_s
       + \E\left(\frac{z(r-z)}{n^3(n-1)}\frac{n(n-1)r^2}{n(n-r) + z(r-z)}\sum_{i=1}^n\sum_{j\neq i} y_{is}y_{js}T_{is}T_{js}\right) \nn &=& 
       \frac{r-1}{n-1}\sigma^2_s
       + \frac{z(r-z)}{n^3(n-1)}\E\left(\frac{n(n-1)r^2}{n(n-r) + z(r-z)}\sum_{i=1}^n\sum_{j\neq i} y_{is}y_{js}T_{is}T_{js}\right) \nn&=&
       \frac{r-1}{n-1}\sigma^2_s +  \frac{z(r-z)}{n^3(n-1)}\sum_{i=1}^n\sum_{j\neq i} y_{is}y_{js} = \var(\hat\mu_{s,\text{HT}}).
     \end{eqnarray}
 
   From~\eqref{jointindicateexp2}, it follows that:
   \begin{eqnarray}
      &&\E\left(\frac{n(n-1)r^2(r-1)}{n^2(r-1) - z(r-z)}\sum_{i=1}^n\sum_{j\neq i} y_{is}y_{jt}T_{is}T_{jt}\right)\nn&=&
      \frac{n(n-1)r^2(r-1)}{n^2(r-1) - z(r-z)}\sum_{i=1}^n\sum_{j\neq i} y_{is}y_{jt}\E(T_{is}T_{jt})\nn&=&
       \frac{n(n-1)r^2(r-1)}{n^2(r-1) - z(r-z)}\sum_{i=1}^n\sum_{j\neq i} y_{is}y_{jt} \frac{n^2(r-1) - z(r-z)}{n(n-1)r^2(r-1)}\nn&=&
       \sum_{i=1}^n\sum_{j\neq i} y_{is}y_{jt}
       \label{midptyisyjt}
   \end{eqnarray}
   Thus, by~\eqref{midptyisyjt}:
   \begin{eqnarray}  
   &&\E\left(\frac{2r^2z(r-z)}{n^4(r-1) - n^2z(r-z)}\sum_{i=1}^n\sum_{j\neq i} y_{is}y_{jt}T_{is}T_{jt}\right) \nn&=&
   \E\left(\frac{2z(r-z)}{n^2}\frac{r^2}{n^2(r-1) - z(r-z)}\sum_{i=1}^n\sum_{j\neq i} y_{is}y_{jt}T_{is}T_{jt}\right) \nn&=&
   \E\left(\frac{2z(r-z)}{n^3(n-1)(r-1)}\frac{n(n-1)r^2(r-1)}{n^2(r-1) - z(r-z)}\sum_{i=1}^n\sum_{j\neq i} y_{is}y_{jt}T_{is}T_{jt}\right) \nn&=&
    \frac{2z(r-z)}{n^3(n-1)(r-1)}  \E\left(\frac{n(n-1)r^2(r-1)}{n^2(r-1) - z(r-z)}\sum_{i=1}^n\sum_{j\neq i} y_{is}y_{jt}T_{is}T_{jt}\right) \nn&=&
     \frac{2z(r-z)}{n^3(n-1)(r-1)}\sum_{i=1}^n\sum_{j\neq i}y_{is}y_{jt}.
   \end{eqnarray}


\begin{thebibliography}{}
  	
  	\bibitem[Abadie and Imbens, 2008]{abadie07}
  	Abadie, A. and Imbens, G. (2008).
  	\newblock Estimation of the conditional variance in paired experiments.
  	\newblock {\em Annales d'Economie et de Statistique}, No. 91-92:175--187.
  	
  	\bibitem[Cochran, 1977]{cochran1977sampling}
  	Cochran, W. (1977).
  	\newblock {\em Sampling techniques}.
  	\newblock Wiley, New York, NY.
  	
  	\bibitem[Fisher, 1926]{fisher26}
  	Fisher, R.~A. (1926).
  	\newblock The arrangement of field experiments.
  	\newblock {\em Journal of the Ministry of Agriculture of Great Britain},
  	33:503--513.
  	
  	\bibitem[Hardy et~al., 1952]{hardyEtal52}
  	Hardy, G., Littlewood, J., and P\'{o}lya, G. (1952).
  	\newblock {\em Inequalities}.
  	\newblock Cambridge University Press, second edition.
  	
  	\bibitem[Holland, 1986]{holland1986statistics}
  	Holland, P.~W. (1986).
  	\newblock Statistics and causal inference.
  	\newblock {\em Journal of the American statistical Association},
  	81(396):945--960.
  	
  	\bibitem[Horvitz and Thompson, 1952]{horvitz1952generalization}
  	Horvitz, D.~G. and Thompson, D.~J. (1952).
  	\newblock A generalization of sampling without replacement from a finite
  	universe.
  	\newblock {\em Journal of the American Statistical Association},
  	47(260):663--685.
  	
  	\bibitem[Imai, 2008]{imai08}
  	Imai, K. (2008).
  	\newblock Variance identification and efficiency analysis in randomized
  	experiments under the matched-pair design.
  	\newblock {\em Statistics in medicine}, 27(24):4857--4873.
  	
  	\bibitem[Imbens, 2011]{imbens11}
  	Imbens, G.~W. (2011).
  	\newblock Experimental design for unit and cluster randomized trials.
  	\newblock Working Paper.
  	
  	\bibitem[Lohr, 1999]{lohr1999sampling}
  	Lohr, S. (1999).
  	\newblock {\em Sampling: Design and Analysis}.
  	\newblock Duxbury Press, Pacific Grove, CA.
  	
  	\bibitem[Miratrix et~al., 2013]{miratrix2012adjusting}
  	Miratrix, L.~W., Sekhon, J.~S., and Yu, B. (2013).
  	\newblock Adjusting treatment effect estimates by post-stratification in
  	randomized experiments.
  	\newblock {\em Journal of the Royal Statistical Society, Series B},
  	75(2):369--396.
  	
  	\bibitem[Neyman, 1935]{neyman35}
  	Neyman, J. (1935).
  	\newblock Statistical problems in agricultural experimentation (with
  	discussion).
  	\newblock {\em Supplement of Journal of the Royal Statistical Society}, 2:107
  	Ð 180.
  	
  	\bibitem[Rubin, 1974]{rubin1974estimating}
  	Rubin, D.~B. (1974).
  	\newblock Estimating causal effects of treatments in randomized and
  	nonrandomized studies.
  	\newblock {\em Journal of Educational Psychology; Journal of Educational
  		Psychology}, 66(5):688.
  	
  	\bibitem[Splawa-Neyman et~al., 1990]{splawa1990application}
  	Splawa-Neyman, J., Dabrowska, D., and Speed, T. (1990).
  	\newblock On the application of probability theory to agricultural experiments.
  	essay on principles. section 9.
  	\newblock {\em Statistical Science}, 5(4):465--472.
  	
  \end{thebibliography}
\end{document}